\documentclass[11pt,twoside]{article}

\usepackage{fullpage}
\usepackage{epsfig}
\usepackage{amsfonts}
\usepackage{amssymb}
\usepackage{amstext}
\usepackage{amsmath}
\usepackage{amsthm}
\usepackage{xspace}
\usepackage{shadow}

\usepackage{times}

\newtheorem{claim}{Claim}

\newtheorem{lemma}[claim]{Lemma}
\newtheorem{theorem}{Theorem}

\newtheorem{fact}[theorem]{Fact}

\makeatletter
\renewcommand{\section}{\@startsection{section}{1}{0pt}{-12pt}{5pt}{\large\bf}}
\renewcommand{\subsection}{\vspace*{-.1in}\@startsection{subsection}{2}{0pt}{-12pt}{-5pt}{\normalsize\bf}}
\renewcommand{\subsubsection}{\vspace*{-.1in}\@startsection{subsubsection}{3}{0pt}{-12pt}{-5pt}{\normalsize\bf}}
\makeatother

\newcommand{\ignore}[1]{}
\newcommand{\opt}{\mathrm{OPT}}
\newcommand{\alg}{\mathrm{ALG}}

\newcommand{\OPT}{\mathrm{OPT}}
\newcommand{\ALG}{\mathrm{ALG}}

\newcommand{\ads}{A}      

\newcommand{\realizationop}[1]{\hat{#1}}   

\newcommand{\Rreq}{\realizationop{\req}}  

\newcommand{\req}{I}     
\newcommand{\reqs}{\req}
\newcommand{\imp}{\req}
\newcommand{\imps}{\req}

\newcommand{\Draws}{\Rreq}  
\newcommand{\draws}[1]{D(#1)}   

\newcommand{\G}{G}       
\newcommand{\FG}{G_f}       

\newcommand{\RG}{\realizationop{\G}}       
\newcommand{\RE}{\realizationop{\E}}       

\newcommand{\FRG}{\realizationop{\G}_f}       

\newcommand{\E}{E}        

\newcommand{\fedges}{E_f}     
\newcommand{\cedges}{E_{\delta}}  

\newcommand{\dist}{{\cal D}}  

\newcommand{\adss}[1]{\ads_{\mathrm{#1}}}

\newcommand{\RS}{\realizationop{S}}   
\newcommand{\RT}{\realizationop{T}}

\newcommand{\todo}[1]{\footnote{\large {\bf TODO: } #1}}

\newcommand{\bibR}{{\cal R}}

\newcommand{\delineated}[1]{\medskip \noindent {\em #1.}}

\newcommand{\MOFF}{M_{\text{OFF}}}
\newcommand{\instance}{\Gamma}

\newcommand{\st}{s \! - \! t}

\newenvironment{proofof}[1]{\noindent{\em Proof of #1:}}{\hfill\qed}

\title{Online Stochastic Matching: Beating $1-{1\over e}$}
\author{
{\large Jon Feldman}
\thanks{Google, Inc., New York, NY. {\tt jonfeld@google.com}}
\and
{\large Aranyak Mehta}
\thanks{Google, Inc., Mountain View, CA. {\tt aranyak@google.com}}
\and
{\large Vahab Mirrokni}
\thanks{Google, Inc., New York, NY. {\tt mirrokni@google.com}}
\and
{\large S. Muthukrishnan}
\thanks{Google, Inc., New York, NY. {\tt muthu@google.com}}
}
\date{\today}

\begin{document}

\maketitle
\begin{abstract}
We study the online stochastic bipartite matching problem, in a form
motivated by display ad allocation on the Internet.  In the online,
but adversarial case, the celebrated result of Karp, Vazirani and
Vazirani gives an approximation ratio of $1-{1\over e} \simeq 0.632$, a
very familiar bound that holds for many online problems; further, the bound is
tight in this case.  In the online, stochastic case
when nodes are drawn repeatedly from a
known distribution, the greedy algorithm matches
this approximation ratio,  but still, no algorithm
is known that beats the $1 - {1\over e}$ bound.

Our main result is a $0.67$-approximation
online algorithm for stochastic bipartite matching,
breaking this $1 - {1\over e}$ barrier.
Furthermore, we show that no online algorithm can produce
a $1-\epsilon$ approximation for an arbitrarily small $\epsilon$ for
this problem.

Our algorithms are based on computing an
optimal offline solution to the expected instance,
and using this solution as a guideline in
the process of online allocation.
We employ a novel application of the idea of the power of
two choices from load balancing:
we compute two disjoint solutions to the expected instance, and
use both of them in the online algorithm in a prescribed preference order.
To identify these two disjoint solutions, we solve a max flow problem
in a boosted flow graph, and then carefully decompose this maximum flow
to two edge-disjoint (near-)matchings.
In addition to guiding the online decision making,
these two offline solutions are used to characterize an upper bound
for the optimum in any scenario.  This is done by
identifying a cut whose value we can bound under the arrival distribution.

At the end, we discuss extensions of our results to more general
bipartite allocations that are important in a display ad application.
\end{abstract}
\thispagestyle{empty}

\newpage
\setcounter{page}{1}

\section{Introduction}

Bipartite matching problems are central in combinatorial optimization with many applications.
Our motivating application is the allocation
of display advertisements on the Internet,\footnote{For details
of this application, see Section~\ref{sec:motivation}.}
and so we will use the language of
this application to define and discuss the problem:

\medskip
\noindent {\em (Online Bipartite Matching)} There is a bipartite graph
$G(\ads,\imp,\E)$ with advertisers $\ads$ and impressions $\imp$, and a set
$E$ of edges between them.  Advertisers in $\ads$ are fixed and
known. Impressions (or requests)
in $\imp$ (along with their incident edges) arrive online. Upon the arrival of an
impression $i \in \imp$, we must assign $i$ to any advertiser $a \in
\ads$ where $(i,a) \in E(G)$.  At all times, the set of assigned
edges must form a matching (that is, no end points coincide). \hfill \qed

\medskip
If the online algorithm knows nothing about $\imp$ or $\E$ beforehand,
and the impressions arrive in an arbitrary order, we have the {\em adversarial model}.
Then, Karp, Vazirani and Vazirani~\cite{KVV} solved this problem
by presenting an online algorithm with
an approximation ratio of $1-1/e \simeq 0.632$, and further showed
that no algorithm can achieve a better ratio.

A different model is the online, stochastic one called
the {\em iid} model, where impressions $i \in \imp$ arrive
online according some {\em known} probability distribution (with
repetition).  In other words, in addition to $\G$, we are given a
probability distribution $\dist$ over the elements of $\imp$.  Our
goal is then to compute a maximum matching on $\RG = (\ads, \Draws,
\RE)$, where $\Draws$ is drawn from $\dist$.\footnote{We give more details on
this model in Section~\ref{sec:prelim}, including a discussion of
different ways to characterize an approximation ratio in this context.}
In this {\em iid} model, the greedy algorithm achieves an
approximation ratio of $1-1/e$~\cite{GM08,AM09}.
Nothing better is known.

Another stochastic model is the {\em random order model}
where we assume that $\imp$ is unknown, but
impressions in $\imp$ arrive in a random order.  This has proved be an
important analytical construct for other problems such as
secretary-type problems where worst cases are inherently difficult. It
is known that in this case even the greedy algorithm has a (tight)
competitive ratio of $1-{1\over
e}$~\cite{GM08}. Further, no deterministic algorithm can
achieve approximation ratio better than $0.75$ and no randomized
algorithm better than $0.83$~\cite{GM08}. Currently the best known
approximation ratio remains $1-1/e$.

Can one beat the $1-1/e$ bound?
We address this main question.

\ignore{
The online ad allocation motivation~\footnote{For a
more complete description of the motivation in ad allocation,
see Section~\ref{sec:motivation}.}  gives us insight into an
alternative model for the online matching problem.
In practice, impressions arrive stochastically according to
some distribution defined by Internet users.
These distributions can be estimated by ad serving
companies using the history of  user activities.
}

\subsection{Our Results and Techniques.}
We present two results for the online stochastic bipartite matching
problem under the {\em iid} model.
\begin{itemize}
\item
We present an algorithm with an approximation factor of $\frac{ 1 -
\frac{2}{e^2} }{ \frac 4 3 - \frac{2}{3e} } \simeq 0.67$, breaking
past the $1-1/e$ bottleneck.  We also show that our analysis is tight,
by providing an example for which our algorithm achieves exactly this
factor.

\item
We show that there is no $1-o(1)$-approximation algorithm for this
problem. Specifically, we show that any online algorithm will be off
by at least $26/27$ (or $\approx .99$ if one requires a family of
instances that grows with $n$).
\end{itemize}

Our algorithms are based on computing an optimal offline solution, and
using it to guide online allocation.
An intuitive approach under this paradigm is to compute a matching $\MOFF$ on the
``expected graph''---that is, the one that would result if all
impressions occurred exactly as many times as expected.
 Thereafter, one can use this matching online,
that is, when node $i \in \imps$ arrives, match it with $a \in \ads$ iff
$(i,a) \in \MOFF$.
One expects this to perform well if the empirical probability of
occurrence of each node $i\in \imps$ is very close to its value in the
distribution.
This can be shown if all $i \in \imps$ occur very frequently using
for example the Chernoff bound.
However in general, many $i \in \imps$ will have very low frequency.
In this paper, we show that
this first attempt achieves (you guessed it) $1 -1/e$, and this is tight.

To get our main result and beat $1-1/e$, we
compute {\em two} disjoint offline solutions and use them as follows:
when a request arrives, we try to assign it based on the first offline
solution, and if that assignment fails, we try the second.
In order to identify these two disjoint offline solutions, we solve a max flow problem
in a boosted flow graph, and then carefully decompose this maximum flow
to two edge-disjoint (near)-matchings.
Other than guiding the online decision making,
these offline solutions are used to characterize an upper bound
for the optimum in each scenario.
This bound
is determined by identifying an appropriate cut in each
scenario that is guided by a cut in the offline solution.
This is the main technical part of the analysis, and we hope this
technique proves useful for analyzing heuristic algorithms
for other stochastic optimization
problems.\footnote{For example, this technique might
be applicable for proving performance guarantees for heuristics
for approximate  dynamic programming problems studied in the
OR literature~\cite{BC99,Bertsekas,FR04,FR06}.}

The idea of using two solutions
is inspired by the idea of power of two choices in online load
balancing~\cite{ABKU99,M01}. Power of two choices
has traditionally meant choosing between two {\em random} choices
for online allocation; in contrast, we use two {\em deterministic} choices, carefully computed offline to
guide online allocation.\footnote{Previously, power of two choices has
been used in various congestion control and load balancing
settings. Our work is a novel
adaptation of this idea to a stochastic bipartite matching setting.}

Our results are somewhat more general as shown in the technical sections, and the problem itself was motivated from an
Internet ad application described later.

\subsection{Other Related Work.}
Our online stochastic matching problem is an
example of online decision making problems studied in the Operations
Research  literature as stochastic approximate dynamic programming
problems~\cite{BC99,Bertsekas,FR04,FR06}.
Several heuristic methods
have been proposed for such problems (e.g., see Rollout algorithms
for stochastic dynamic programming in ~\cite{Bertsekas}), but we are not aware of any
rigorous analysis of the performance of the heuristics.
Recently other online stochastic
combinatorial optimization problems like
Steiner tree and set cover problems have been studied in the {\em iid}
model~\cite{Gupta-FOCS,Gupta-SODA};
one can achieve  an approximation factor
better than the best bound for the adversarial
online variant.

A related ad allocation problem is the {\em Adwords assignment}
problem~\cite{msvv} that was motivated by sponsored search
auctions. When modeled as an online bipartite assignment problem,
here, each edge has a {\em weight}, and there is a {\em budget} on each ad
representing the upper bound on the total weight of edges that may be
assigned to it.
In the offline setting, this problem is NP-Hard, and
several approximations have been
designed~\cite{chakrabarty2008aba,srinivasan2008baf,azar2008iaa}. For
the online setting, it is typically assumed that every weight is very
small compared to the corresponding budget, in which case there exist
$1-1/e$ factor online
algorithms~\cite{msvv,buchbinder-jain-naor,GM08,AM09}. Recently, it has
been brought to our attention that an online
algorithm~\cite{devanur-hayes} gives a $1-\epsilon$-approximation, for
any $\epsilon$, for
Adwords assignment when $opt$ is larger than
$O({n^2\over \epsilon^3})$ times each bid
in the iid and random permutation  models.
Thus, technically, our problem is different from
their problem in
two ways: the edges are unweighted (making it easier), but
$\opt$ is not necessarily much larger than each bid
(making it harder -- in the bipartite graph case, $\opt$ can
be $O(n)$). Moreover, our offline problem is solvable in polynomial
time, and we show that no $1-\epsilon$-approximation can be achieved for
our problem for some fixed $\epsilon$. In fact, their algorithm, along with other previously
studied algorithms (e.g, algorithms based on
greedy, greedy bid-scaling, and primal-dual techniques)
does not achieve a factor better than $1-{1\over e}$ for our problem,
and we beat $1-{1\over e}$ factor using a different technique.
An interesting related model for combining stochastic-based
and online solutions for the Adwords problem
is considered in ~\cite{MNS}, but their approach does
not give an improved approximation algorithm for the
{\em iid} model.

\ignore{
In terms of modeling, the display ad business is easier to model,
since currently, display ads are not sold via auctions, and prices are
the same for different impressions of an advertiser (so we do not need
to worry about the underlying auction pricing schemes). Differing
values of ad slots to different advertisers is handled exogenously via
sales contracts, and the online problem is just to assign edges to
meet the contracted sales.
}

\subsection{Applied Motivation: Display Ad Allocation.}\label{sec:motivation}

Our motivation is in part applied and
arises from allocation of ``display ads'' on the
Internet. Here is a high level view.  Websites have
multiple pages (e.g., sports, real estate, etc), and several
slots where they can display ads (say an image or video or a
block of text).  Each user who views one of these pages is shown
ads, i.e., the ads get what is called an ``impression.''  Advertisers pay the website
per impression and buy them (typically in lots of one thousand)
ahead of time, often specifying a subset of pages on which they would
like their ad to appear, or a type of user they wish to target.
All such sales are entered into an ad
delivery system (ADS).

Since the ADS serves ads on the same web pages from day to day, they
have an idea of the traffic that occurs on these websites.   While there are
inaccuracies and indeed it is nearly impossible to forecast the number of
viewers of a webpage in the future, it is standard industry practice to use
these estimates at the time of selling inventory to various
advertisers (to judge whether a new sale can be
accommodated).

When a user visits one of the pages, the ADS determines the set of
eligible ads for that slot, and selects an ad to be shown.  Since not
all ads are suitable for each page or slot, we have an online (in two
senses of the word) bipartite matching scenario.  The ADS would like
to maximize the number of impressions that are filled with ads in
order to satisfy their contracts, and thus maximize their revenue.

The underlying problem is an online bipartite matching problem in the
{\em iid}
model.  Each $i \in \imp$ is an ``impression type,'' which
may represent a particular web page, or even a cross product of
targeting criteria (location, demographic, etc.).  Edges $(a, i)$ then
capture the fact that advertiser $a$ was interested in an impression
of type $i$.  Using past traffic data, the ADS defines $e_i$ to be the
typical number of impressions they get of type $i$.  Then, the
distribution $\dist$ over $\imp$ is given by $\mathrm{Pr}[i^*] = {e_{i^*}
\over \sum_{i} e_{i}}$.

In contrast to sponsored search,
the display ad business is easier to model,
since currently, display ads are not sold via auctions, and prices are
the same for different impressions of an advertiser (so we do not need
to worry about the underlying auction pricing schemes). Differing
values of ad slots to different advertisers is handled exogenously via
sales contracts, and the online problem is just to assign edges to
meet the contracted sales.
Still, we note that there are many aspects of online ad serving that deserve
a richer model than the one we give here, and indeed there is more
work to be done in this area.  For example,
the ADS may want to maximize the value of the contracts fulfilled,
rather than the total number of impressions, or may want to maximize
some notion of quality of ads served.  One extension that we address
is frequency capping, which we discuss in the conclusion. As such,
display ad selection problems are solved routinely by ADSs, and any
insights or solutions we develop for our problem are likely to be
useful in practice.

\section{Preliminaries}
\label{sec:prelim}

Consider the following online stochastic matching problem in the i.i.d model:
We are given a bipartite graph $\G = (\ads, \reqs, \E)$ over advertisers
$\ads$ and impression types $\req$.
Let  $k = | \ads |$ and $m = |\reqs |$.
We are also given, for each impression type $i \in \req$, an integer number
$e_i$ of impressions we expect to see.  Let $n = \sum_{i \in \req}
e_i$.  We use $\dist$ to denote the distribution over $\req$ defined by $\text{Pr}[i] = e_i/n$.

An instance $\instance=(\G, \dist, n)$ of the {\em online stochastic
matching} problem is as follows: We are given offline access to $\G$
and the distribution $\dist$.  Online, $n$ i.i.d. draws of impressions
$i \sim \dist$ arrive, and we must immediately assign ad impression
$i$ to some advertiser $a$ where $(a, i) \in \E$, or not assign $i$ at
all.  Each advertiser $a \in \ads$ may only be assigned at most
once\footnote{All results in this paper hold for a more general case
that each advertiser $a$ has a capacity $c_a$ and advertiser $a$ can
be assigned at most $c_a$ times.  This more general case can be
reduced easily to the degree one case by repeating each node $a$ $c_a$
number of times in the instance.}.  Our goal is to
assign arriving impressions to advertisers and maximize the total
number of assigned impressions. In the following, we will formally
define the objective function of the algorithm.

Let $\draws{i}$ be the set of draws of impression type $i$ that arrive during the run of the algorithm.
We let a scenario $\Draws = \cup_{i \in \req} \draws{i}$ be the set of impressions.
Let $\RG(\Draws)$ be the
``realization'' graph, i.e., with node sets $\ads$ and $\Draws$, and
edges $\RE = \{ (a, i') : (a, i) \in E, i' \in \draws{i}  \}$.

Given an instance $\instance=(\G,{\dist}, n)$
of the online matching problem,
we wish an algorithm ALG for which
for any instance $\instance$ of the online matching problem,
with high probability ${\alg(\Draws)\over \opt(\Draws)}\ge \alpha$. In this case,
we say that the algorithm achieves approximation factor $\alpha$ with
high probability.
One could also study weaker notions of approximation, namely
${E[\alg(\Draws)] \over E[\opt(\Draws)]}$ (the approximation factor in expectation), or
$E[{\alg(\Draws)\over \opt(\Draws)}]$ (the expected approximation factor).
Note that if one proves a high-probability factor of $\alpha$, it implies
an approximation factor in expectation, and an expected approximation
factor of at least $\alpha - o(1)$.

\subsection{Balls in Bins.}

In this section we characterize two useful extensions of the standard
balls-in-bins problem, where we are interested in the distribution of
certain functions of the bins.  We characterize the expectations of
these functions, and use Azuma's inequality on appropriately defined
Doob's Martingales to establish concentration results as needed.
In particular, we will use the following facts. The proofs are left to the appendix.

\begin{fact}
\label{thm:bibsubset}
Suppose $n$ balls are thrown into $n$ bins, i.i.d. with uniform probability over the bins.
Let $B$ be a particular subset of the bins, and $S$ be a random variable that equals the number of bins from $B$ with at least one ball.
With probability at least $1 - 2 e^{-\epsilon n / 2}$, for any $\epsilon > 0$, we have
$|B| (1 - {1 \over e}) - \epsilon n \leq \: S \: \leq |B| (1 - {1 \over e} + {1 \over e n}) + \epsilon n$
\end{fact}


\begin{fact}
\label{thm:bib}
Suppose $n$ balls are thrown into $n$ bins, i.i.d. with uniform probability over the bins.
Let $B_1, B_2, \dots, B_\ell$ be ordered sequences of bins, each of size
$c$, where no bin is in more than $d$ such sequences.
Fix some arbitrary subset $\bibR \subseteq \{1, \dots, c\}$.
We say that a bin sequence $B_a = (b_1, \dots, b_c)$ is ``satisfied'' if
(i) at least one of its bins $b_i$ with $i \not \in \bibR$ has at least one ball in it; or,
(ii) at least one of its bins $b_i$ with $i \in \bibR$ has at least two balls in it.
Let $S$ be a random variable that equals the number of satisfied bin sequences.
With probability at least $1 - 2 e^{-\epsilon^2 n / 2}$, we have
$S \geq \ell (1 - \frac{2^{|\bibR|}}{e^c}) - \epsilon d n - \frac{2^{|\bibR|} c^2}{e^c} {\ell \over n - c^2}$.
\end{fact}

\section{Hardness}
In this section, we show that the expected approximation factor of
every (randomized) online algorithm is bounded strictly away from 1.

Consider the 6-cycle $\G$ defined by $\ads = \{a, b, c\}$, $\reqs = \{x,y, z\}$,
and $\E=\{(x,a)$, $(y,a)$, $(y,b)$, $(z,b)$, $(z,c)$, $(x,c)\}$.
The distribution $\dist$ is the uniform distribution
$(1/3,1/3,1/3)$ on $\reqs$, and $n=3$. We show that no (randomized)
algorithm can achieve an expected approximation factor better than
$26/27$ on this instance.
Without loss of generality (from the symmetry of the 6-cycle),
assume that the first impression to arrive is $x$ and
that it gets assigned to advertiser $a$. Now, if the next two arrivals
are both of impression $y$, then any algorithm will only be able to
assign one of these. The optimal assignment for the scenario $(x,y,y)$ is
to assign $x$ to $c$, and the two $y$ impressions to $a$ and $b$.
Since the probability of $(x,y,y)$ is $1/9$, the
expected approximation factor is at most $(1/9) (2/3) + (8/9) 1 =
26/27$.

To get a family of instances on which no algorithm can do better than
a constant bounded away from 1, we will have to construct
an instance consisting of a large number $k$ copies of 6-cycles.
Using this idea, we can prove the following theorem.
The details of the proof are left to the appendix.

\begin{theorem}
\label{thm:hardness}
There is an instance of the online
stochastic matching problem in which no algorithm can achieve an
expected approximation factor better than $26 \over 27$. Moreover,
there exists a
family of instances with $n\rightarrow \infty$
for which no algorithm can achieve an
expected approximation of $1-o(1)$.
\end{theorem}

\section{Offline Algorithms for Online Matching}
\label{sec:alg}

In this section, we present our improved online algorithms guided by
offline solutions.  Before stating the improved approximation result,
we ``warm up'' with a simple, natural algorithm that uses the idea of
computing an offline solution to ``guide'' our online choices.  This
algorithm will only achieve a $1 - {1 \over e}$-approximation (which
is tight). The proof of this part illustrates the framework
we will use in the second section to beat $1-{1 \over e}$; however we
will need a new idea to achieve this---namely, the use of a {\em second} offline solution.

\subsection{``Suggested Matching'' Algorithm: a $1-{1\over e}$-Approximation.}\label{sec:onematchingalg}
The {\em suggested matching} algorithm is a first attempt
at the approach of using an offline solution for online matching.  In
this algorithm, we simply find a maximum matching in the graph we ``expect''
to arrive, then restrict our online choices to this matching.

\delineated{Offline Algorithm} We will describe this algorithm more formally in terms of
the standard characterization
of $b$-matching as a max-flow problem, since we will later use this flow graph
explicitly to bound $\opt$.
Given an instance $\instance=(\G(\ads, \req, \E),\dist, n)$ of the
problem, we will find a max-flow in a graph $\FG$ constructed from $G$
as follows:
define a new source node $s$ and an edge $(s, a)$ with capacity $1$ to
all $a \in \ads$, direct all edges in $E$ from $\ads$ to $\req$, and
add a sink node $t$ with edges $(i, t)$ from all $i \in \req$ with
capacity $e_i$.
Let $f_{ai} \in \{0, 1\}$ be the flow on edge $(a, i)$ in this max
flow (since all the capacities are integers, we may assume that the
resulting flow is integral~\cite{schrijver}).  For ease of notation,
we say $f_{ai} = 0$ if edge $(a,i) \not \in E$.

\delineated{Online Algorithm} When an impression $i' \in \draws{i}$ arrives online, we choose a random
ad $a'$ according to the distribution defined by the flow; i.e., the
probability of choosing $a'$ is $f_{a'i} \over
e_i$.
(Note that if $\sum_a f_{ai} < e_i$ there is some probability that no
$a'$ is chosen.)
If $a'$ is already taken, we do not match $i$ to any
ad.\footnote{Clearly, making an arbitrary available match is always as
good (and in some cases better) than doing nothing; we present the
algorithm this way for ease of presentation.}

\newcommand{\mads}{\ads^*}

\delineated{Bounding $\alg$}
The performance of this algorithm is easily characterized with high
probability in terms of the computed max-flow.
Define $F_a = \sum_i f_{ai}$, and note that $F_a \in \{0, 1\}$; this indicates whether ad $a$ was chosen in the max flow.
Let $\mads = \{ a \in \ads : F_a = 1\}$.
When an impression $i \in \Draws$ arrives online, a particular ad $a : f_{a, i} = 1$ has probability $1 / e_i$ of being chosen by the online algorithm;
since each impression $i$ has probability $e_i / n$ of arriving, we conclude that
each $a \in \mads$ has probability $1/n$ of being chosen by the online algorithm upon each arrival.
Thus, to bound the total number of ads chosen we have a balls-in-bins
problem with $n$ balls and $n$ bins, and we are interested in
lower-bounding the number of bins (among a subset of size $|\mads|$) that have at least one ball.  Applying concentration results for balls-in-bins (Fact~\ref{thm:bibsubset}), we get that with probability $1 - e^{-\Omega(n)}$,
$
\ALG \geq ( 1 - {1 \over e} ) |\mads| - \epsilon n.
$

\delineated{Bounding $\opt$}
To bound the optimal solution, we will construct a cut in the
realization graph $\RG = (\ads, \Draws, \RE)$ using a min-cut of $\FG$
(constructed using the max-flow found by the algorithm) as a
``guide.''
Let $(S, T)$ be a min $\st$ cut in the graph $\FG$ using the canonical
``reachability'' cut in $\FG$; i.e., $S$ is defined as the set of nodes
reachable from $s$ using paths in the residual graph after sending the
flow $f$ found by the algorithm.  This is always a min-cut.\cite{schrijver}
Let $\ads_S = \ads \cap S$ and define $\ads_T$, $\req_S$ and $\req_T$
similarly.

We claim that there are no edges in $\E$ from $\ads_S$ to $\req_T$;
suppose there is such an edge $(a, i)$.  Then, $a$ much be reachable
from $s$ since $a \in S$, but $i$ must not be reachable since $i \in
T$.  This implies that there is no residual capacity along $(a,i)$;
i.e., $f_{ai} = 1$.  However this also implies that there is no residual
capacity along $(s, a)$ since $(s,a)$ is the only edge entering $a$
and it has capacity $1$, and that there is no other flow leaving $a$.
This implies that $a$ is not reachable in the residual graph, a
contradiction.
Thus the only edges in the cut $(S,T)$ are from $s$ to $A_T$ (capacity
1) and from $i \in \req_S$ to $t$ (capacity $e_i$).
We may conclude using max-flow min-cut that
$|\mads| = \sum_a F_a = |\ads_T| + \sum_{i \in \req_S} e_i.$

Now consider the ``realization'' graph $\RG = (\ads, \Draws, \RE)$,
and define a max-flow instance $\FRG$ whose solution has size equal to
the maximum matching in $\RG$; i.e., create a source $s$ with edges to
all $a \in \ads$, direct edges of $\RE$ toward $\Draws$, and create a
sink $t$ with edges from all $i' \in \Draws$.  Set the capacity of every
edge to one.
Note that any $\st$ cut in $\FRG$ is a bound on $\OPT$.

We define an $\st$ cut in $(\RS, \RT)$ in $\FRG$ as follows.
Let
$\Draws_S = \cup_{i \in \req_S} \draws{i}$ and
$\Draws_T = \cup_{i \in \req_T} \draws{i}$.
Define $\RS = \ads_S \cup \Draws_S$ and $\RT = \ads_T \cup \Draws_T$.
Note that since there are no edges from $\ads_S$ to $\req_T$ in $\FG$, there are also no edges
from $\ads_S$ to $\Draws_T$ in $\FRG$.
Thus the size of the cut $(\RS, \RT)$ is equal to $|\Draws_S| + |\ads_T|$.
An online impression ends up in the set $\Draws_S$ with probability
$\sum_{i \in \req_S} e_i / n$, independent of the other impressions.
Using a Chernoff bound, we can conclude that for any $\epsilon > 0$, with probability $1- e^{-\Omega(n)}$ (over the scenarios), the size of the cut (and therefore $\OPT$) obeys
$
\OPT \leq |\ads_T| +  \sum_{i \in \req_S} e_i + \epsilon n = |\mads| + \epsilon n.
$

\delineated{Tightness of the Analysis}
Consider a special case of the online matching problem
$\instance(\G, \dist, n)$ where $e_i=1$ for each $i\in \req$ and
the underlying graph $\G$ is a complete
bipartite graph. The algorithm will find a perfect
matching between $\req$ and $\ads$, and so
each ad is matched with probability at least
$1-{1\over e}$. Using Fact~\ref{thm:bibsubset}, the algorithm achieves $\approx (1-{1\over e})n$ with
high probability. However, the optimum is $n$. Therefore:


\begin{theorem}
\label{thm:mainfrac}
The approximation factor of the {\em suggested matching} algorithm is $1-{1\over e}$
with high probability,
and this is tight, even in expectation.
\end{theorem}

\subsection{``Two Suggested Matchings'' (TSM) Algorithm: Beating $1-{1\over e}$.}\label{mainresult}
To improve upon the {\em suggested matching} algorithm, we will
instead use {\em two} disjoint (near-)matchings to guide our
online algorithm.  To find these matchings, we boost the capacities of the flow graph and then decompose the resulting solution into disjoint solutions.
The second solution allows to to break the $1- {1 \over e}$ barrier and prove:

\begin{theorem}
\label{thm:mainuniform}
For any $\epsilon > 0$,
with probability at least $1 - e^{-\Omega(n)}$,
as long as $\OPT = \Omega(n)$,
the {\em two suggested matchings} algorithm achieves approximation ratio
$$
\frac{\ALG}{\OPT} - \epsilon
\; \geq \; \alpha := \frac{ 1 - \tfrac{2}{e^2} }{ \tfrac 4 3 - \tfrac{2}{3e} }
\; \approx \; 0.67029
\; > \; 1 - \tfrac 1 e.
$$
Moreover, this ratio is tight; specifically,
there is a family of instances for which the {\em two suggested matchings} algorithm
has expected approximation factor at most $\alpha + \epsilon$.
\end{theorem}


Throughout the section, until the final proof of
Theorem~\ref{thm:mainuniform}, we assume $e_i = 1$ for all $i \in
\req$, which also implies $m = n$.  Extending to integer $e_i$ is a
simple reduction to this case.

\subsubsection{The TSM Algorithm.}

In this algorithm, we construct a {\em boosted flow graph} $\FG$,
built from $\G$ in the standard
reduction of matching to max-flow; i.e., create a source $s$ with
edges to all $a \in \ads$, direct the edges of $G$ towards nodes in
$\req$, and create a sink $t$ with edges from all $i \in \req$.
However, we set the capacities of the edges differently than in the
max-flow reduction: (i) Edges $(s, a)$ from the source get capacity 2,
(ii) edges $(a, i) \in \E$ get capacity 1, and
(iii) edges $(i, t)$ from $\reqs$ to $t$ get capacity 2.

We find a max-flow in this graph from $s$ to $t$.  Since all the
capacities are integers, we may assume that the resulting flow is
integral~\cite{schrijver}.  Let $\fedges$ be the set of edges $(a,i)
\subseteq \E$ with non-zero flow on them, which must be unit flow.  Since the
capacities of edges $(s, a)$ and $(i, t)$ are all 2, we know that the
graph induced by $\fedges$ is a collection of
paths and cycles.  Using this structure, we assign colors blue and red to the edges of $\fedges$ as follows:
\begin{itemize}\addtolength{\itemsep}{-0.5\baselineskip}
\item
Color the cycle edges alternating blue and red.
\item
Color the edges of the odd-length paths alternating blue and red, with more blue than red.
\item
For the even-length paths that start and end with nodes $a \in \ads$, alternate blue and red.
\item
For the even-length paths that start and end with impressions $i \in \req$, color the first two edges blue, and then alternate red, blue, red, blue, etc., ending in blue.
\end{itemize}
Note that all $i \in \req$ are incident to either no colored edges, one blue edge, or
a blue and a red edge.

The TSM algorithm for serving online ad
impressions is simple: For each $i \in \imp$, the first time $i$ arrives try the blue edge; the second time $i$ arrives try the red edge.  More
formally, for all $i \in \imp$ maintain a count $x_i$ of the number of impressions $i' \in \draws{i}$ that have arrived so far.
When $i' \in \draws{i}$ arrives:
if $x_i = 0$, set $a'$ to be the ad along $i$'s blue edge (if $i$ has a blue edge);
if $x_i = 1$, set $a'$ to be the ad along $i$'s red edge (if $i$ has a red edge).
Now assign $i$ to $a'$ if $a'$ is unassigned.
If this $a'$ is already assigned, or if $x_i > 1$, do not make an assignment.\footnote{A slight improvement to this algorithm is to try to match along the red edge if matching along the blue edge fails; we do not make use of this in the analysis so we leave it out for clarity.}

\newcommand{\biggif}{}

\subsubsection{Performance of the TSM Algorithm.}
To analyze the performance of this algorithm, we first derive a lower
bound on the number of ads assigned during the run of the algorithm.
We do so in terms of the incidence pattern of the different ads with respect to the edges $\fedges$.
Specifically, let $\adss{BR}$ be the ads that are incident to a blue
and a red edge, and $\adss{B}$ be the ads that are incident to only a
blue edge.  Similarly define $\adss{BB}$ and $\adss{R}$.  We have
\begin{eqnarray}
\label{eq:flow1}
|\fedges| = 2 \adss{BR} + 2 \adss{BB} + \adss{B} + \adss{R}.
\end{eqnarray}
Consider some $a \in \adss{B}$ with blue edge $(a, i)$.  The event
that $a$ is ever chosen is exactly the event that some $i' \in
\draws{i}$ is ever drawn from $D$, since then we will choose $a$ (and
no other impression will choose $a$).
Since $e_i = 1$, this is exactly the probability that a particular
bin
is non-empty in a balls-in-bins problem with $n$ balls (the online impressions), and $m=n$ bins (the impression types $\req$).
Applying Fact~\ref{thm:bibsubset},
we get that with high probability
the number of ads chosen from $\adss{B}$  is at least
$
|\adss{B}| (1 - \frac{1}{e}) - \epsilon n.
$
Now consider some $a \in \adss{BR}$ with blue edge $(a, i_b)$ and red
edge $(a, i_r)$.
If $|\draws{i_b}| \geq 1$, or if $|\draws{i_r}| \geq 2$, then $a$ will definitely be chosen.
Thus we can apply Fact~\eqref{thm:bib} with $n$ balls, $m=n$ bins, $c = 2$,
bin sequences equal to the neighborhood sets of $\adss{BR}$ along the blue and red edges (ordered blue, red),
$d = 2$ (since each impression is incident to at most 2 edges of $\fedges$),
and $\bibR$ set to the second (red) bin of the bin sequence.
We conclude that
with high probability, the number of ads chosen from $\adss{BR}$ is at least
$
|\adss{BR}|  ( 1 - \frac{2}{e^2}  ) - \epsilon n.
$
Similar reasoning gives bounds with coefficients of $( 1 - \frac{1}{e^2} )$ for $\adss{BB}$ and $( 1 - \frac 2 e)$ for $\adss{R}$.
We may conclude that with high probability (over the scenarios),
$
\ALG \geq
\biggif ( 1 - \frac{1}{e^2} \biggif ) |\adss{BB}| +
\biggif ( 1 - \frac{2}{e^2} \biggif ) |\adss{BR}| +
\biggif ( 1 - \frac 1 e \biggif ) |\adss{B}| +
\biggif ( 1 - \frac 2 e \biggif ) |\adss{R}| - 4 \epsilon n.
$
Note that since $|\ads_B| \geq |\ads_R|$, we can also assert
\begin{eqnarray}
\label{eq:algbound}
\ALG \geq
\biggif ( 1 - \tfrac{1}{e^2} \biggif ) \adss{BB} +
\biggif ( 1 - \tfrac{2}{e^2} \biggif ) \adss{BR} +
\biggif ( 1 - \tfrac{3}{2e} \biggif ) (\adss{B} + \adss{R}) - 4 \epsilon n.
\end{eqnarray}

\subsubsection{Bound on the optimal solution.}
Let $(S, T)$ be a particular min $\st$ cut of the flow graph $\FG$
defined as follows.
First start with the canonical ``reachability''
min $\st$ cut of the flow graph $\FG$, where $S$ is defined as the set
of nodes reachable from $s$ in the residual graph $\FG$ left after
finding the max-flow $\fedges$.
Then, we do a small bit of ``surgery'' to this cut:  for all $i \in \req \cap T$, if $i$ is incident to more than one
$a \in \ads \cap S$, we move $i$ over to $S$.  Note that this does not
increase the value of the cut, since we save at least 2 for the two
edges from $\ads \cap S$, and pay exactly 2 for the edge $(i, t)$.
Let $\ads_S = \ads \cap S$, and define $\ads_T$, $\reqs_S$ and $\reqs_T$ similarly.
Let $\cedges$ be the
set of
 edges $(a,i) \in E$ that cross the cut (from $\ads_S$ to
$\reqs_T$).

Some observations:
{\bf (i)} We have $\cedges \subseteq \fedges$, since otherwise, if some $(a, i)
\in \cedges$ has no flow across it, then $i$ would be reachable from
$s$, and would not be in the set $T$.  (And we did not introduce any such
edges in our surgery.)
{\bf (ii)}
All $i \in \req_T$ have at most one incident edge in $\cedges$ (follows from the surgery).
{\bf (iii)}
All $a \in \ads_S$ have at most one incident edge in $\cedges$. To see
this, suppose it had two such edges (it cannot have more than 2 since
$\cedges \subseteq \fedges$).  Then, since $a$ is reachable from $s$
(since it is in $S$), it must have either residual capacity from $s$
directly, or residual capacity from $\req_S$; but it cannot have
either, since $(s, a)$ is saturated and both flow edges from $a$ go to
$\req_T$.

Let $\ads_\delta$, $\req_\delta$ be the ads and impressions,
respectively, that are incident to edges in $\cedges$.  We may
conclude from the observations above that the graph $(\ads_\delta,
\req_\delta, \cedges)$ induced by $\cedges$ is a matching.
The min-cut of $\FG$ is made up of the edges $\cedges$, the $|\ads_T|$ edges
from $s$ to $\ads_T$ (with capacity 2), and the $|\reqs_S|$ edges from
$\reqs_S$ to $t$ (also capacity 2).  Thus, by max-flow-min-cut, we have
\begin{eqnarray}
\label{eq:flow2}
|\fedges| = 2 (|\ads_T| + |\reqs_S|) + |\cedges|.
\end{eqnarray}
We are interested in bounding the value of the optimal matching in the
realization graph $\RG = (\ads, \Draws, \RE)$.  To do this, we will
use the min-cut $(S, T)$ of the graph $\FG$ as a ``guide'' to
construct a (not necessarily min) cut in a flow graph built from $\RG$, and prove a
high-probability bound on the size of this cut.

More precisely, we let $\FRG$ be a directed version of $\RG$, constructed as before
with a source and a sink, and edges corresponding to $\RG$; but now we put capacity 1 on all edges.
Note that any $\st$ cut in this graph constitutes an upper bound on
$\OPT$, the maximum matching in $\RG$.
We construct such a cut $(\RS, \RT)$ as follows.
We let $\Draws_S = \cup_{i \in \req_S} \draws{i}$ and $\Draws_T =
\cup_{i \in \req_T} \draws{i}$.
For the ads, we will
use almost the same partition $(\ads_S, \ads_T)$ as in $\FG$ but we will perform
some ``surgery'' on this partition as well.
Let $\ads^*_\delta \subseteq \ads_\delta$ be the set of ads $a \in \ads$
that are incident (in
$\FRG$) to some $i' \in \draws{i} \subseteq \Draws_T$.
Note that $i \in \req_\delta$ and $(a, i) \in \cedges$.
We set
$\RS = \Draws_S \cup (\ads_S \setminus \ads^*_\delta)$ and
$\RT = \Draws_T \cup \ads_T \cup \ads^*_\delta$.

Now we will measure the size of the cut $(\RS, \RT)$ in $\FG$.  We pay 1 for
each $a \in \Draws_S$, $i \in \ads_T$ and $i \in \ads^*_\delta$.  But note
that there are no edges in $\FG$ from $\ads \cap \RS$ to $\Draws_T$, since we got rid of them in our surgery.
Thus we have
$
\OPT \leq |\Draws_S| + |\ads_T| + |\ads^*_\delta|.
$

Using a Chernoff bound, with probability $1- e^{-\Omega(n)}$
we have
$|\Draws_S| \leq |\req_S| + \epsilon n$ for any $\epsilon > 0$.
To bound $|\ads^*_\delta|$, consider some $a \in \ads_\delta$, and the
impression $i \in \req_\delta$ along the edge $(a, i)$ in the matching
$(\ads_\delta, \req_\delta, \cedges)$.
The ad $a$ appears in $\ads^*_\delta$ iff impression $i$ is drawn during
the run of the algorithm.
Thus we have a balls-in-bins problem with $n$ balls, $m = n$ bins,
uniform bin probabilities and a bin subset of size $|\ads^*_\delta|$, and
we are concerned with an upper bound on the number of bins in that
subset that get at least one ball.  Using Fact~\ref{thm:bibsubset} we
may conclude that with high probability $|\ads^*_\delta| \leq (1 -
\frac 1 e) |\cedges| + \epsilon n + O(1)$.

Summarizing the previous arguments, we get, for any $\epsilon > 0$, with probability $1- e^{-\Omega(n)}$,
$
\OPT \leq |\req_S| + |\ads_T| + (1 - \frac 1 e)|\cedges| + \epsilon n.
$
Applying Equations~\eqref{eq:flow2} then~\eqref{eq:flow1}, we get
\begin{eqnarray}
\OPT
& \leq & \tfrac 1 2 |\fedges| + (\tfrac 1 2 - \tfrac 1 e)|\cedges| + \epsilon n \nonumber \\
& = & |\adss{BR}| + |\adss{BB}| + \tfrac 1 2 (|\adss{B}| + |\adss{R}|) + (\tfrac 1 2 - \tfrac 1 e)|\cedges| + \epsilon n
\label{eq:optbound}
\end{eqnarray}
In order to use this bound on $\OPT$ together with the bound on $\alg$
in Equation~\ref{eq:algbound}, we must bound the size of $\cedges$ in
terms of the sets $\adss{BR}$, $\adss{BB}$, $\adss{B}$ and $\adss{R}$.
The following lemma takes a deeper look at the two matchings
constructed by the algorithm, and their relationship to the min-cut
$(S, T)$ in $\FG$, in order to achieve this bound.

\begin{lemma}
\label{lemma:deltabound}$|\cedges| \leq \frac 2 3  |\adss{BR}| + \frac 4 3  |\adss{BB}| +  |\adss{B}| +  \frac 1 3 |\adss{R}|.$
\end{lemma}

\begin{proof}
It suffices to show that the inequality holds for every connected
component (path or cycles) of the graph induced by $\fedges$.
We thus assume notationally that the graph induced by $\fedges$ consists
of a single such connected component.

Consider an arbitrary pair of edges $(a_1, i_1), (a_2, i_2) \in
\cedges \subseteq \fedges$.  Since the edges of $\cedges$ are
independent, $(a_1, i_1)$ and $(a_2, i_2)$ cannot occur consecutively
in this component (path or cycle); we claim further that $(a_1, i_1)$ and $(a_2,
i_2)$ must have at least two edges between them.  Suppose not, then
wlog $(a_2, i_1)$ is in the component; but since $a_2 \in \ads_S$
and $i_1 \in \req_T$ (by the definition of $\cedges$) we must have
$(a_2, i_1) \in \cedges$, contradicting the fact that the edges
$\cedges$ are independent.

If the component is a cycle of length $k$, we can use the reasoning above to conclude
that there are at most $\lfloor \frac k 3 \rfloor$ edges of $\cedges$ in
the cycle.  The ads in the cycle are all in $\adss{BR}$ and there are
exactly $\frac k 2$ of them.  Thus $|\cedges| \leq \frac 2 3 |\adss{BR}|$, which implies the inequality.

If the component is a path, we can conclude that $|\cedges| \leq \lceil \frac k 3 \rceil$
by the reasoning above---the worst case is when the path starts and
ends in a $\cedges$ edge.
We have three cases for this path, depending on the parity of its length, and (in the case of even-length paths) whether it starts and ends in $\ads$ or $\req$.
\begin{itemize}\addtolength{\itemsep}{-0.5\baselineskip}
\item
For odd paths of length $k$, by construction of the edge colors, we have one ad in $\adss{B}$ and
$\frac{k-1}{2}$ ads in $\adss{BR}$.  Thus
$
|\cedges|
\leq \biggif \lceil \frac k 3 \biggif \rceil
= \biggif \lceil \frac{2 |\adss{BR}|}{3} + \frac 1 3 \biggif \rceil
\leq \frac{2}{3} |\adss{BR}| + 1
= \frac 2 3 |\adss{BR}| + |\adss{B}|.
$

\item
For even paths of length $k$ that start and end with ads, we have
$|\adss{B}| = 1$, $|\adss{R}| = 1$ and $|\adss{BR}| = \frac k 2 - 1$.  Thus
$
|\cedges|
\leq \lceil \frac k 3 \rceil
= \lceil \frac{2 |\adss{BR}|}{3} + \frac 2 3 \rceil
$.  We bound this using a case analysis on $|\adss{BR}| \mod{3}$, as follows:
{\bf (i)}
If $|\adss{BR}| \equiv 0 \mod{3}$ then we get $|\cedges| \leq \frac 2 3 |\adss{BR}| + 1$.
{\bf (ii)}
If $|\adss{BR}| \equiv 1 \mod{3}$ then we get $|\cedges| \leq \frac 2 3 |\adss{BR}| + \frac 4 3$.
{\bf (iii)}
If $|\adss{BR}| \equiv 2 \mod{3}$ then we get $|\cedges| \leq \frac 2 3 |\adss{BR}| + \frac 2 3$.
In all cases this is less than $\frac 2 3 |\adss{BR}| + \frac 4 3 =
\frac 2 3 |\adss{BR}| + |\adss{B}| + \frac 1 3 |\adss{R}|$.

\item
For even length paths that start and end in
impressions, we have $|\adss{BB}| = 1$ and $|\adss{BR}| = \frac k 2 - 1$.
As in the previous case we can say
$
|\cedges|
\leq  \lceil \frac k 3 \rceil
= \lceil \frac{2 |\adss{BR}|}{3} + \frac 2 3 \rceil
$,
and reason by the same case analysis that
$|\cedges|$ is at most $\frac 2 3 \adss{BR} + \frac 4 3$.  This is equal to
$\frac 2 3 \adss{BR} + \frac 4 3 \adss{BB}$.
\end{itemize}
\vspace{-.65cm}
\end{proof}

\subsubsection{Proof of Theorem~\ref{thm:mainuniform}.}
We first prove the approximation ratio for $e_i = 1$.
The bounds in equations~\eqref{eq:algbound} and~\eqref{eq:optbound}
each hold with probability $1 - e^{-\Omega(n)}$, and so
using a union bound they both hold with probability $1 - e^{-\Omega(n)}$.
Using Lemma~\ref{lemma:deltabound} (ignoring the $\frac 1 3$ in
front of the $\ads_R$) and Equation~\eqref{eq:optbound}, we get
$$
\OPT \leq (\frac 4 3 - \frac{2}{3e}) \adss{BR} + (\frac 5 3 - \frac{4}{3e}) \adss{BB} + (1 - \frac 1 e)(\adss{B} + \adss{R}) + \epsilon' n.
$$
Since $\OPT = \Omega(n)$, we can choose $\epsilon'$ small enough such that when we apply
Equation~\eqref{eq:algbound} (also using $\epsilon'$) we may conclude
$$
\frac{\ALG}{\OPT} + \epsilon \geq \min \bigg \{
\frac{ 1 - \frac{1}{e^2} }{ \frac 5 3 - \frac{4}{3e} },
\frac{ 1 - \frac{2}{e^2} }{ \frac 4 3 - \frac{2}{3e} },
\frac{ 1 - \frac{3}{2e} }{ 1 - \frac 1 e} \bigg \}
 = \min \{ .735 \dots , .670 \dots , .709 \dots \}
 = \frac{ 1 - \frac{2}{e^2} }{ \frac 4 3 - \frac{2}{3e} }
\approx .670.
$$

The tightness of this analysis is proved in Section~\ref{sec:tight}.
For arbitrary integer $e_i$, we give a reduction to the case $e_i = 1$.
Given a set of instance $\instance=(\G, \dist, n)$, we reduce to a
new instance $\instance' =(\G', \dist', n)$
with $e'_i = 1$ by making $e_i$ copies of each impression type $i$.
Then, when an impression of type $i$ arrives online, ``name'' it randomly
according to one of its copies.  The resulting distribution $\dist'$ is uniform over the impression types $\req'$ in the new instance.

Let $\Draws$ be the impressions that are drawn from $\dist$ in one run
of the algorithm, and let $\Draws'$ be the resulting draws from $\dist'$.  By the arguments above, we achieve the desired bound on $\ALG /
\OPT'$ with high probability, where $\OPT'$ is with respect to
$\Draws'$; however we have $\OPT' = \OPT$, since the realization
graphs $\RG = (\ads, \Draws', \RE')$ and $\RG' = (\ads, \Draws, \RE)$
are in fact the same graph.

\subsubsection{Tightness of the analysis for the TSM  Algorithm.} \label{sec:tight}
In this section we demonstrate a family of instances for which the TSM
algorithm achieves a factor no better than $\tfrac{1-2/e^2}{4/3 - 2/(3e)}$,
thus showing that the analysis in Section~\ref{sec:alg} is tight.

The family is parameterized by $n$, which is the number of
advertisers, the number of impression types, as well as the number of
impression arrivals. We shall take $n$ to be a multiple of 4.  The set
$\ads$ of advertisers consists of the following parts: a set $K$ of
size $\frac{n}{4}$ and, for $i \in [1, \frac{n}{4}]$, advertisers
$\{u_i, v_i, w_i\}$.  The set $\reqs$ of impressions consists of the
following parts: a set $L$ of size $\frac{n}{4}$ and, for $i \in [1,
\frac{n}{4}]$, impressions $\{x_i, y_i, z_i\}$. Define $U = \{u_i:~i
\in [1,n/4]\}$, and similarly, $V, W, X, Y, Z$. Thus $\ads = K \cup U
\cup V \cup W$ and $\reqs = L \cup X \cup Y \cup Z$. Draws are from
the uniform distribution on $\reqs$.

The edges $\E$ are as follows:
{\bf (i)}
For $i \in [1, \frac{n}{4}]$, the 6-cycle
$\{u_i-x_i-v_i-y_i-w_i-z_i-u_i\}$,
{\bf (ii)}
a complete bipartite graph between $K$ and $X$, and
{\bf (iii)}
a complete bipartite graph between $L$ and $W$.

We now describe the max-flow and min-cut in $\FG$ found during the
algorithm. The only edges with (unit) flow are the edges of the
6-cycles, i.e., for $i \in [1, \frac{n}{4}]$,
$\{u_i-x_i-v_i-y_i-w_i-z_i-u_i\}$. Thus all vertices in $U, V, W, X, Y$
and $Z$ have a flow of 2 each, and the vertices in $K$ and $L$ have a
flow of $0$. The reachability cut $(S,T)$ obtained from this flow has $S
= K \cup X \cup U \cup V \cup \{s\}$ (where $s$ is the source
vertex). The flow and the cut
both have size $\tfrac{3n}{2}$.
Using Fact~\ref{thm:bibsubset}, one can easily check that the
algorithm achieves the total matching size of $\tfrac{3n}{4}
(1-\tfrac{2}{e^2})$ with high probability.

The following assignment
can be made with high probability, and is a lower bound on OPT.
{\bf (i)}
With high probability there will be $\frac{n}{4}$ draws of
impressions from $X$ (with repeats). These can be matched to
the $\frac{n}{4}$ advertisers in $K$ (in any order).
{\bf (ii)}
With high probability there will be $\frac{n}{4}$ draws of
impressions from $L$. These can be matched to the $\frac{n}{4}$
advertisers in $W$.
{\bf (iii)}
With high probability there will be $(1-\frac{1}{e})\frac{n}{4}$
unique draws of impressions from $Y$ (counting each $y_i$ only once,
even if it is drawn multiple times). For every such $y_i$, its first
draw is matched to $v_i$, and the repeat draws of $y_i$ are left
unmatched. Similarly, with high probability, there are
$(1-\frac{1}{e})\frac{n}{4}$ unique draws of $z_i$'s, and these are
matched to the corresponding $u_i$'s.
Thus, this assignment has size $\frac{n}{4} + \frac{n}{4} + (1-\frac{1}{e}) \frac{n}{2} =
n(1- \frac{1}{2e})$. This means that the TSM algorithm cannot achieve a
factor better than $({1-\tfrac{2}{e^2}}) / ({\tfrac{4}{3} - \tfrac{2}{3e}})$.

\section{Concluding Remarks}


\noindent
{\bf Applying the insights to the display ads application.}
The approach of using the offline solution to
allocate ads online may be quite useful in practice because while one can invest some time offline to find the guiding
solutions, the online allocation has to be done very quickly in this application.  One can use this approach to model
other objective functions such as fairness in quality of ad slots assigned to ads, which may be solvable
offline with some computational effort.
As an example, we elaborate on the extension of our algorithm to the following problem.
In the display ads business, advertisers have ``frequency caps;'' i.e., they do not want the same
user to see their ad more than some fixed (constant) number of times. We can extend our approach here to get
a $1-1/e$ approximation as shown in the appendix.

\medskip
{\bf \noindent Generalizing the algorithm.}
One can generalize the two-matching algorithm to a ``$k$-matching'' algorithm by computing $k$ matchings
instead of 2 matchings, and then using them online in a prescribed order. We can easily show that if the underlying expected graph $\G$
admits $k$ edge-disjoint perfect matchings, the approximation factor of such an algorithm is $1-{2\over e^2}\simeq 0.72$
and $1-{5\over e^3}\simeq 0.75$ for $k=2$ and $k=3$ respectively, however
for $k=3$, we do not know how to generalize our result for to graphs.
One natural question left open by this work is what constant $c(k)$
is achieved by extending to $k$ matchings, where $.67 \leq c(k) \leq .99$.

\medskip
{\bf \noindent Fractional version.}
A theoretical version of online stochastic matching problem that may be of interest is the case in which
$e_i$'s are not necessarily integers, but arbitrary rational numbers.
We observe the analysis of the  ``one suggested matching'' algorithm can be
generalized to this case, but do not know how to generalize the analysis of the ``two suggested matchings'' algorithm.
The details are in the appendix.

\section*{Acknowledgements.}
We thank Ciamac Moallemi and Nicole Immorlica for pointing us to related work.


\appendix

\section{Balls in Bins}
In this part, we prove the concentration facts we used throughout the paper.

\bigskip
{\bf \noindent Fact \ref{thm:bibsubset}.}
{\em
Suppose $n$ balls are thrown into $n$ bins, i.i.d. with uniform probability over the bins.
Let $B$ be a particular subset of the bins, and $S$ be a random variable that equals the number of bins from $B$ with at least one ball.
With probability at least $1 - 2 e^{-\epsilon n / 2}$, for any $\epsilon > 0$, we have
$$
|B| (1 - {1 \over e}) - \epsilon n
\leq \: S \: \leq
|B| (1 - {1 \over e} + {1 \over e n}) + \epsilon n
$$
}
\bigskip

\begin{proofof}{Fact~\ref{thm:bibsubset}}
We have $E[S] = \sum_{a \in B} 1 - (1 - {1 \over n})^n$ and so using standard identities, we obtain
$$
\sum_{a \in B} 1 - e^{-1}
\leq \: E[S] \: \leq
\sum_{a \in B} 1 - e^{-1} (1 - {1 \over n}).
$$
Since $S$, as a function of the placements of the $n$ balls, satisfies the Lipschitz condition, we may apply
 Azuma's inequality to the Doob Martingale and obtain
$$
\mathrm{Pr}[|S - E[S]| \geq \epsilon n] \leq 2 e^{-\epsilon^2 n / 2}.
$$
\end{proofof}

\bigskip
\bigskip

{\bf \noindent Fact \ref{thm:bib}.}
{\em
Suppose $n$ balls are thrown into $n$ bins, i.i.d. with uniform probability over the bins.
Let $B_1, B_2, \dots, B_\ell$ be ordered sequences of bins, each of size
$c$, where no bin is in more than $d$ such sequences.
Fix some arbitrary subset $\bibR \subseteq \{1, \dots, c\}$.
We say that a bin sequence $B_a = (b_1, \dots, b_c)$ is ``satisfied'' if
\begin{itemize}
\item
at least one of its bins $b_i$ with $i \not \in \bibR$ has at least one ball in it; or,
\item
at least one of its bins $b_i$ with $i \in \bibR$ has at least two balls in it.
\end{itemize}
Let $S$ be a random variable that equals the number of satisfied bin sequences.
With probability at least $1 - 2 e^{-\epsilon^2 n / 2}$, we have
$$
S \geq \ell (1 - \frac{2^{|\bibR|}}{e^c}) - \epsilon d n - \frac{2^{|\bibR|} c^2}{e^c} {\ell \over n - c^2}
$$
}

\begin{proofof}{Fact~\ref{thm:bib}}
First, we claim
$$
E[S] \geq \ell \bigg ( 1 - \frac{2^{|\bibR|}}{e^c} \bigg ( 1 + {c^2 \over {n - c^2}} \bigg ) \bigg )
$$
To see this, fix some bin sequence $B_a$.  The probability that $B_a$ is not satisfied is
$$
\sum_{\bibR' \subseteq \bibR} \binom{n}{|\bibR'|} |\bibR'|! \; n^{-|\bibR'|} \bigg ( 1-\frac{c}{n} \bigg )^{n - |\bibR'|}
\leq \sum_{\bibR' \subseteq \bibR} \bigg ( 1-\frac{c}{n} \bigg )^{n - c}
\leq \frac{2^{|\bibR|}}{e^c} \bigg ( 1 + {c^2 \over n - c^2} \bigg ).
$$
The bound on $E[S]$ follows by linearity of expectation.
Now, consider $S$ as a function of the placements of the $n$ balls.
Moving one ball can affect $S$ by at most $d$, since each bin is in at
most $d$ sequences.  Thus we may apply Azuma's inequality and obtain,
for all $\epsilon > 0$,
$$
\mathrm{Pr}[|S - E[S]| \geq \epsilon d n] \leq 2 e^{-\epsilon^2 n / 2}.
$$
\end{proofof}

\section{Details of the proof for Hardness Result}
Consider the instance which consists
of a large number $k$ copies of  6-cycles, the uniform distribution
on the union of the impressions, and $n = 3k$.  Let $\gamma_1, \gamma_2,
\gamma_3$ and $\gamma_+$ be the fraction of the cycles that receive $1, 2,
3$ and more than $3$ impressions, respectively.  We have (using a simple application of Azuma's inequality) that with high probability,
\begin{eqnarray*}
\gamma_1 & = & 3k            {1 \over k}   (1 - {3 \over 3k})^{3k - 1} \simeq {3 \over e^3},\\
\gamma_2 & = & \binom{3k}{2} {1 \over k^2} (1 - {3 \over 3k})^{3k - 2} \simeq {9 \over 2 e^3},\\
\gamma_3 & = & \binom{3k}{3} {1 \over k^3} (1 - {3 \over 3k})^{3k - 3} \simeq {27 \over 6e ^3},\\
\gamma_+ & = & 1 - \gamma_1 - \gamma_2 - \gamma_3
\end{eqnarray*}
For cycles that receive 1 or 2 impressions, we can assume that both
$\alg$ and $\opt$ match 1 or 2 ads, respectively.  As we are
upper-bounding $E[\alg/\opt]$, we may assume that on cycles that
receive more than 3 impressions, both $\alg$ and $\opt$ achieve 3
matches, which maximizes the contribution of these cycles to the ratio $\alg/\opt$.

For cycles that receive exactly 3 impressions, we have the same situation
as in the single cycle above.  We assume wlog that $x$ arrives first
and is matched to ad $a$.  If the other two impressions are also both
$x$, then both $\alg$ and $\opt$ match two ads ($a$ and $c$) for this
cycle.  If the other two impressions are both $y$, we have that $\alg$
matches at most two ads but $\opt$ matches three.  In all other
scenarios, we assume that both $\alg$ and $\opt$ match three ads.  By
a Chernoff bound, with high probability the scenarios $(x, x, x)$ and
$(x, y, y)$ each happen $\simeq \gamma_3 k / 9$ times.

Summarizing, we have argued that with high probability,
\begin{eqnarray*}
{\alg \over \opt} \; \simeq \;
\frac{\gamma_1 + 2 \gamma_2 + 3 \gamma_+ + (2 \cdot {2 \over 9} + 3 \cdot {7 \over 9}) \gamma_3}
{\gamma_1 + 2 \gamma_2 + 3 \gamma_+ + (2 \cdot {1 \over 9} + 3 \cdot {8 \over 9}) \gamma_3}
\; \simeq \;
\frac{6 e^3 - 23}
{6 e^3 - 22}
\; \approx \;
.9898.
\end{eqnarray*}
This establishes Theorem~\ref{thm:hardness}.

\section{Non-integral Impression Arrival Rates}\label{fractional}
One natural extension of the online stochastic matching problem is the case in which
$e_i$'s are not necessarily integers, but arbitrary rational numbers.
We observe that the ``Suggested Matching'' algorithm, with $1-{1\over e}$-approximation factor, easily generalizes
to this case, as follows: instead of computing a maximum matching, we can
compute a maximum flow, $f$, on the corresponding flow graph,  and upon the arrival
of an impression $i$, assign the impression $i$ to an ad $a$ with probability $f_{ia}$, i.e.,
proportional to the fractional edge from $i$ to $a$.  Given
the total fraction $F_a$ on each ad $a$, we can argue that this
algorithm achieves value $\sum_{a\in A} (1-e^{-F_a})$ with high probability.
Moreover, one can show that optimum is at most $\sum_{a\in A} F_a$ with high
probability. As a result, the approximation factor of the algorithm can
be captured by the ratio $\sum_{a\in A} (1-e^{-F_a})\over \sum_{a\in A} F_a$
where $0\le F_a\le 1$ for all $a\in A$.
Since $0 \leq F_a \leq 1$, we can characterize this bound as
the solution to the following mathematical program:

\begin{eqnarray*}
\min & \sum_{a\in A} (1-e^{-F_a}) &  \\
s.t.& \sum_{a\in A} F_a = 1 & \\
& 0\le F_a\le 1 & \forall a\in A
\end{eqnarray*}

This mathematical program can be solved analytically. Consider the
vector $\Phi$ of values $F_1,\ldots,F_{|A|}$ in nonincreasing order of
$F$'s, and let $f(\Phi)=\sum_{a\in A} (1-e^{-F_a})$. For any vectors
$\Phi_1$ and $\Phi_2$ subject to $||\Phi||_1 = 1$, if $\Phi_1$
majorizes $\Phi_2$, then clearly $f(\Phi_1)\geq f(\Phi_2)$. Since the
uniform vector $\bar{\Phi}=[1/|A|,\ldots,1/|A|]$ is majorized by all
the vectors, $f(\bar{\Phi})=|A|(1-e^{-1/|A|})$ is the minimum value
attainable. When $|A|=1$, $f(\bar{\Phi})=1-1/e$. We derive that
${df(\bar{\Phi})\over d|A|}= 1-e^{-1/|A|}+|A|(-e^{-1/|A|} \times
1/|A|^2) = 1-e^{-1/|A|}-e^{-1/|A|}/|A|$.  Now,
$1-e^{-1/|A|}-e^{-1/|A|}/|A| \geq 0$ because multiplying by
$e^{1/|A|}$, we get $e^{1/|A|} \geq 1+1/|A|$ which follows from the
Maclaurin series expansion of $e^x$. Thus, $df(\bar{\Phi})/d|A| \geq
0$ and this implies that the solution to the mathematical program is
attained at $1-{1\over e}$.  Therefore, the approximation factor of
the algorithm is equal to $1-{1\over e}$ with high probability.

Generalizing the TSM algorithm to non-integer $e_i$s needs a proper decomposition of
the flow on the corresponding flow graph to two edge-disjoint flows each with large
values. Unlike the integral case, such edge-disjoint decomposition is not possible
for the non-integer $e_i$'s and one need to exploit other ideas to
analyze the algorithm. We leave this as an open question.

\section{Frequency Capping}\label{freqcapping}
A useful generalization of the online matching problem that is
well-motivated by the ad allocation application is when the
advertisers have ``frequency caps;'' i.e., they do not want the same
user to see their ad more than some fixed (constant) number of times.
We can regard the user as a ``feature'' of the
impression; i.e., that an ``impression'' $i$ as we've used it in this
paper is in fact a pair $\langle i, u \rangle$, where $u$ is a
particular user, and we have a distribution that gives us $e_{\langle
i, u \rangle}$, the expected number of impressions of each type from
each user.
Also as part of the input, we are given, for each advertiser $a$, a
total number of impressions $d_a$ and a cap $c$ per user.
We could also regard these caps as operating as impression limits on
other features, e.g., demographic or geographic.

Our $1 - {1 \over e}$-approximation algorithm (the {\em suggested
matching} algorithm) from Section~\ref{sec:onematchingalg} is easily
extended to this generalization of the problem.  Here we give a sketch
of this extension.  For the algorithm, we simply make another layer
$U$ of nodes in our max-flow computation, with one node $\langle a, u
\rangle$ for each (advertiser, user) pair.  We make edges from each $a
\in \ads$ to this layer with capacity $c$, and set the capacity of the
edge edge $(s, a)$ to $d_a$.  The algorithm proceeds as before, and
one can easily show with the same argument that the number of
impressions matched is $\simeq F (1 - 1/e)$, where $F$ is the value of
the flow.  Then, by reasoning about the min-cut in this graph, with
some simple reasoning about where this new layer sits in the min-cut,
one can still show that $\OPT$ is bounded by $F$ with high
probability, giving the desired approximation ratio.

Interestingly, it is more challenging to generalize the TSM
algorithm. Setting the capacities to $2 d_a$ and $2$, respectively, of
the top and mid-layer edges does not work as desired, since then the
flow could be spread among more than $d_a$ nodes in the middle layer.

\end{document}